\documentclass{article}
\usepackage[a4paper,top=3cm,bottom=3.5cm,left=3.2cm,right=3.2cm]{geometry}

\usepackage[utf8]{inputenc}
\usepackage{authblk}
\usepackage{graphicx}
\usepackage{subfigure}
\usepackage{amssymb}
\usepackage{lineno}
\usepackage[english]{babel}
\usepackage{palatino, url}
\usepackage{amsmath}

\usepackage{listings}
\usepackage{graphicx}
\usepackage[rightcaption]{sidecap}
\usepackage{textcomp}
\usepackage{array} 
\usepackage{booktabs} 
\usepackage{multirow} 
\usepackage{caption} 
\usepackage{floatflt}
\usepackage{rotating} 
\usepackage{scrextend}
\usepackage{multirow}
\usepackage[numbers]{natbib}
\usepackage{soul}
\usepackage{mathtools}

\newtheorem{prop}{Proposition}

\newtheorem{definition}{Definition}

\newtheorem{proof}{Proof}

\title{Resilience of the autocatalytic feedback loop for gene regulation}

\author[1]{Daniele Proverbio}
\author[1]{Giulia Giordano}
\affil[1]{\small Department of Industrial Engineering, University of Trento, Trento 38123, IT }

\begin{document}

\maketitle

\begin{abstract}
Gene expression in response to stimuli is regulated by transcription factors (TFs) through feedback loop motifs, aimed at maintaining the desired TF concentration despite uncertainties and perturbations. In this work, we consider a stochastic model of the positive gene autoregulating feedback loop and we probabilistically quantify its resilience, \textit{i.e.}, its ability to preserve the equilibrium associated with a prescribed concentration of TFs, and the corresponding basin of attraction, in the presence of noise. We show that the formation of larger oligomers, corresponding to larger Hill coefficients of the regulation function, and thus to sharper non-linearities, improves the system resilience, even close to critical concentrations of TFs. We also explore a complementary definition of resilience that can be assessed within a stochastic formulation relying on the Fokker-Planck equation. Our formal results are accompanied by numerical simulations.
\end{abstract}

\section{Introduction and Motivation}\label{sec:Intro}


Regulating gene expression in response to stimuli is essential for protein production and cell survival. 
Assessing the \textit{resilience} of regulation pathways, \textit{i.e.}, their ability to withstand random perturbations and preserve crucial cell activities to survive and thrive, sheds light onto key biological mechanisms that sustain life and evolution, and enables the design of biomolecular circuits in synthetic biology.

Gene expression pathways, which can be successfully modelled using network motifs and control loops among responsive elements \cite{alon2019introduction,milo2002network}, are fundamental to process external and internal stimuli and accordingly regulate crucial cell functions, such as enzymatic activity \cite{ferrell2013feedback} or changes in gene expression during cell fate decisions \cite{huang2007bifurcation}; they are also essential building blocks for circuits in synthetic biology. 

The positive autoregulating feedback loop motif, where a transcription factor (TF) increases its own production, is a simple and powerful model of gene expression regulation that captures essential mechanisms observed in living cells \cite{wu2009origin} to guarantee robust adaptation and homeostasis \cite{drengstig2012robust}.
The model can exhibit multistability \cite{angeli2004detection}: admitting multiple stable equilibria is evolutionarily advantageous, as it allows to rapidly switch between concentrations of TFs, thereby responding rapidly and reliably to changes in environmental or physiological conditions. For instance, the system modelling the expression of $\beta$\textit{-galactosidase} in \textit{E. coli} is bistable \cite{ozbudak2004multistability}: when lactose becomes available in the environment, the response pathway drives the concentration of the \textit{lac operon} genetic inducer up to a critical threshold, yielding a sudden transition from low (``off'' state, which saves energy) to high (``on'') concentrations of the $\beta$\textit{-galactosidase} enzyme, which helps metabolise the nutrient and allows to feast on lactose before other bacteria, thus enabling survival. Then, as soon as nutrients deplete, another transition brings the system back to the ``off'' state, to avoid wasting resources. 
To efficiently manage gene translation and exploit nutrients to the fullest, transitions between stable equilibria must be tightly controlled and little sensitive to random deviations, environmental disturbances and intrinsic noise \cite{lestas2008noise} affecting cells. Understanding the mechanisms that improve the system resilience in relation to its ``off'' and ``on'' states helps unveil how cells thrive in uncertain environments and inform the development of synthetic pathways.

We consider a model for positive autoregulating feedback control of transcription and we investigate its resilience, formally defined as a probabilistic quantification of its ability to preserve a prescribed attractor and its corresponding basin of attraction in spite of stochastic noise \cite{MTNS2024}. In particular, we assess the equilibria of the nominal deterministic system and their stability properties: the system is bistable for suitable parameter ranges and undergoes fold bifurcations. Then, we employ a stochastic approach to quantify its resilience, around its asymptotically stable equilibria, in the presence of additive Gaussian white noise. Our analysis identifies factors that help the system cope with stochastic perturbations and avoid undesired random transitions between stable equilibria, thus approximately maintaining the desired TF concentrations, and illustrates how cells can thrive in noisy conditions.


\begin{figure}
  \begin{minipage}[c]{0.5\columnwidth}
    \includegraphics[width=0.86\textwidth]{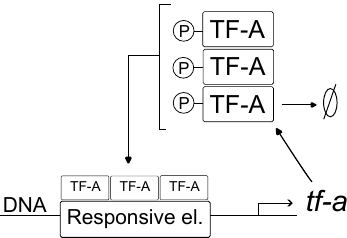}
  \end{minipage}\hfill
  \begin{minipage}[c]{0.5\columnwidth}
    \caption{Positive genetic autoregulating feedback loop \cite{smolen1998frequency}. The transcription factor TF-A, produced by gene \textit{tf-a}, forms an oligomer ($n=3$) that, when phosphorylated (P), increases the transcription rate of \textit{tf-a}, by promoting responsive DNA sequences. TF-A degrades at a constant rate.} \label{fig:scheme}
  \end{minipage}
\end{figure}

\section{Positive Autoregulating Feedback Model}

The positive autoregulating feedback loop model for gene regulation introduced in \cite{smolen1998frequency} is visualised in Fig.~\ref{fig:scheme}. The single activator TF-A belongs to a pathway mediating cellular responses to stimuli and regulates its own transcription \cite{knight1992}.
If the concentration $y \in [0,\infty)$ of TF-A is negligible, transcription of the \textit{tf-a} gene occurs at a basal rate $r_a>0$. TF-A degrades following first-order kinetics, at a constant rate $\zeta>0$.
Also, TF-A can form oligomers having concentration $y^n$, where $n \in \mathbb{N}$ represents how many monomers form the oligomer, which bind to responsive elements and, when phosphorylated, promote the transcription of \textit{tf-a}; phosphorylation can be further regulated by external signals. We assume that binding processes are relatively rapid and close to equilibrium, so that the resulting increase in the transcription rate is captured by a monotonically increasing Hill function, with Hill coefficient $n$ and dissociation constant $K>0$ of oligomers from the responsive elements, which saturates to a maximal rate $\alpha>0$.
All parameters are dimensionless, and the model does not explicitly consider the translation of mRNA into proteins.
In a deterministic setting, the resulting ordinary differential equation (ODE) is
$\tfrac{dy(\tau)}{d\tau} = \alpha \frac{y(\tau)^n}{K + y(\tau)^n} - \zeta y(\tau) + \tilde{r}$.
Considering the new variable $x \doteq y / K^{\frac{1}{n}} \in [0, \infty)$ and rescaling time as $t \doteq \tau/\zeta$ allows us to rewrite the ODE as 
\begin{equation}
    \frac{dx(t)}{dt} = f(x(t)) = a \frac{x(t)^n}{1 + x(t)^n} - x(t) + r,
    \label{eq:ODE_det}
\end{equation}
where $a=\alpha/(\zeta K^{\frac{1}{n}})>0$, $r=\tilde{r}/(\zeta K^{\frac{1}{n}})>0$, and the saturating Hill function $\frac{x^n}{1 + x^n}$, for $n \to \infty$, converges to the Heaviside step function $\Theta(x - 1)$. Model \eqref{eq:ODE_det} with $r=0$ has been thoroughly analysed in \cite{cacace2010discrete}.

We assess the system resilience through the lens of the rigorous formal definitions introduced in \cite{MTNS2024} for a family of ODE systems consisting of stochastic perturbations of a nominal deterministic system, aimed at probabilistically quantifying its ability to preserve a prescribed attractor $A$ (e.g., an asymptotically stable equilibrium) and the corresponding basin of attraction $B(A)$ in spite of noise.
Our \textit{system family} is $\mathcal{F}= \{G_\lambda\}_{\lambda\in \mathcal{I}}$, with
$\mathcal{I}=[0, \hat \lambda)$ and $\hat \lambda>0$, where \eqref{eq:ODE_det} is the \textit{nominal} system $G_{\lambda_0}$, with $\lambda_0 = 0$, whereas the generic system $G_\lambda \in \mathcal{F}$ is described by the stochastic differential equation (SDE)
\begin{equation}\label{eq:ODE_stoch}
\tfrac{dx(t)}{dt} = f(x(t)) +  \lambda \eta(t)  = \tfrac{a x(t)^n}{1 + x(t)^n} - x(t) + r + \lambda \eta(t),
\end{equation}
where $\eta(t)$ is uncorrelated white noise with mean $\langle \eta \rangle = 0$, variance $\sigma_\eta^2=1$
and intensity $\lambda \in \mathcal{I}$. For biological consistency, we assume the noise is such that positivity of the system is preserved: $x(t) \geq 0$ for all $t$. As observed in \cite{proverbio2022buffering}, oligomers that require cooperative binding of more TF-A monomers, leading to a larger Hill coefficient $n$, allow for a better suppression of noise propagation close to the system equilibria. However, no formal results are available for the resilience of system \eqref{eq:ODE_stoch} in relation to its asymptotically stable equilibria, which we investigate in this work.

\section{Bistability of the Deterministic Model}
\label{sec:det:mod}

\begin{figure*}[tb]
  \centering
	\includegraphics[width=0.9\textwidth]{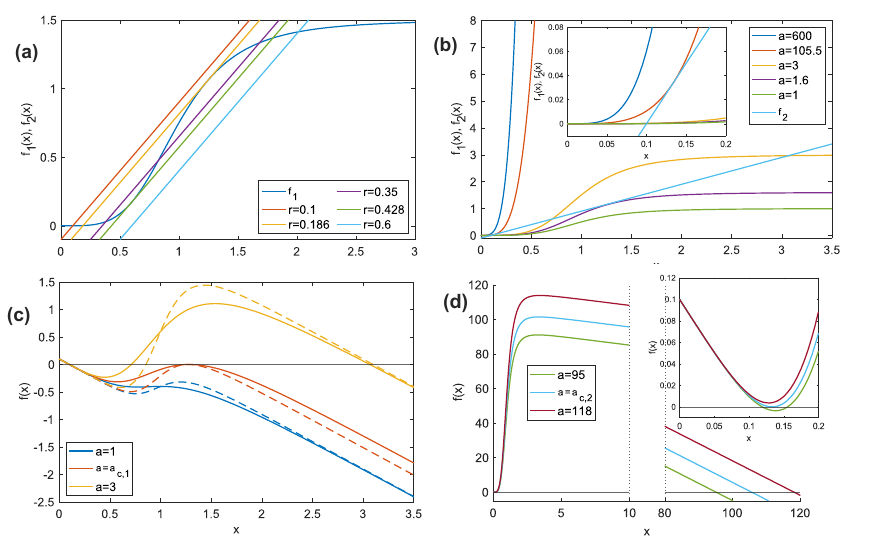}
  \caption{\textbf{(a, b)} The intersections of functions $f_1(x)=a x^n / (1 + x^n)$ and $f_2(x)=x-r$ can be either one, two or three, and correspond to the equilibria of system \eqref{eq:ODE_det}; $\tfrac{dx}{dt} = f(x)=f_1(x)-f_2(x)$ is positive when $f_1(x)>f_2(x)$, zero when $f_1(x)=f_2(x)$, and negative when $f_1(x)<f_2(x)$. \textbf{(a)} $f_1$ with $n=4$ and $a=1.5$; $f_2$ with $r \in \{0.1, 0.186, 0.35, 0.428, 0.6\}$. The intersections are two when $f_2$ is tangent to $f_1$, for $r=r_{c,1} \approx 0.186$ and $r=r_{c,2} \approx 0.428$, and three for $r_{c,1} < r < r_{c,2}$. The intersection is unique for $0 < r < r_{c,1}$ and $r > r_{c,2}$. \textbf{(b)} $f_1$ with $n=4$ and $a \in \{1, 1.6, 3, 105.5, 600\}$; $f_2$ with $r=0.1$. The intersections are two when $f_2$ is tangent to $f_1$, for $a=a_{c,1} \approx 1.6$ and $a=a_{c,2} \approx 105.5$, and three for $a_{c,1} < a < a_{c,2}$. The intersection is unique for $0 < a < a_{c,1}$ and $a > a_{c,2}$. \textbf{(c, d)} Zeros of $f(x)$ with $r=0.1$ for different values of $n$ and of $a$. \textbf{(c)} For $a=a_{c,1}(n)$, $f(x)$ is tangent to the $x$-axis at some point $\bar x(a_{c,1}(n))$: a fold bifurcation occurs that causes a transition from two to three equilibria. Solid lines: $n=4$, with $a_{c,1}(n) \approx 1.6$. Dashed lines: $n=7$, with $a_{c,1}(n) \approx 1.4$. \textbf{(d)} With $n=7$, for $a=a_{c,2} \approx 105.5$, $f(x)$ is tangent to the $x$-axis at some point $\bar x(a_{c,2}(n))$: a fold bifurcation occurs that causes a transition from three to two equilibria.}
\label{fig:intersections}
\end{figure*}

To analyse the qualitative behaviour of \eqref{eq:ODE_det}, we define functions $f_1(x) \doteq a \frac{x^n}{1 + x^n}$ and $f_2(x)\doteq x-r$, such that $f(x)=f_1(x)-f_2(x)$.
We consider $x \in [0,a+r]$, which is an invariant set for \eqref{eq:ODE_det}. In fact, \eqref{eq:ODE_det} is a positive system ($x(0)\geq 0$ implies $x(t) \geq 0$ for all $t>0$, since $\tfrac{dx}{dt} >0$ when $x=0$)  and $\tfrac{dx}{dt} < 0$ when $x \geq a+r$, because $\sup_{x\in[0,\infty)} f_1(x)=a$.

The equilibria of system \eqref{eq:ODE_det} are the intersections of the Hill function $f_1(x)$ and the line $f_2(x)$, as shown in Figs.~\ref{fig:intersections}a,b.

If $n=1$, $f_1$ is a Michaelis-Menten function, and this guarantees structural stability \cite{BLANCHINI2023110683}: the system admits a unique equilibrium $\bar x_1 = \frac{a+r-1+\sqrt{(a+r-1)^2+4r}}{2}$, which is structurally globally asymptotically stable, because $\tfrac{dx}{dt} >0$ for $0 \leq x < \bar x_1$ and $\tfrac{dx}{dt} <0$ for $x > \bar x_1$, for all possible choices of $a>0$ and $r>0$. We next consider $n \geq 2$. Figs.~\ref{fig:intersections}c,d show $f(x)$ for a fixed value of $r$ and different values of $n \geq 2$ and $a$; the system equilibria, corresponding to the zeros of $f(x)$, can be either one (globally asymptotically stable), or two (one asymptotically stable and one unstable), or three (two asymptotically stable and one unstable); their stability can be assessed based on the sign of $f(x)$ in the intervals delimited by the equilibria.


\begin{prop}\label{prop:equilibria_stability}
Given $a>0$, $r > 0$ and $n \geq 2$, system \eqref{eq:ODE_det} admits at most three equilibria $\bar x_i$. If the equilibrium $\bar x_1>0$ is unique, it is globally asymptotically stable. If there are two equilibria, with $0 < \bar x_1 < \bar x_2$, one is asymptotically stable, while the other is unstable. If there are three equilibria, with $0 < \bar x_1 < \bar x_2 < \bar x_3$, the equilibria $\bar x_1$ and $\bar x_3$ are asymptotically stable, while $\bar x_2$ is unstable.
$\hfill\square$
\end{prop}
\begin{proof}
Computing the equilibrium values $\bar x_i$ for which $f(\bar x_i)=0$ requires finding the roots of the polynomial $p(x) = -x^{n+1} + x^n(a+r)-x+r$, which in general is not possible analytically. In view of Descartes' rule of signs, since there are three sign changes between consecutive nonzero coefficients, the polynomial has either one or three positive roots, counted with their multiplicity. This corresponds to three alternative scenarios: one equilibrium, two equilibria (of which one is associated with two coincident roots of $p(x)$), or three equilibria.
As exemplified in Figs.~\ref{fig:intersections}a,b, showing the intersections of $f_1$ and $f_2$ for various parameter choices, all the scenarios are possible. Once all other parameters are fixed, the values $r_{c,1} < r_{c,2}$ of $r$ (or equivalently the values $a_{c,1} < a_{c,2}$ of $a$) for which $f_2$ is tangent to $f_1$ can be found by solving the nonlinear system
\begin{equation}\label{eq:equilibriumtangent}
\begin{cases}
f_1(x)=f_2(x)\\
\tfrac{df_1}{dx}(x)=\tfrac{df_2}{dx}(x)
\end{cases}    
\end{equation}
Fig.~\ref{fig:intersections}a shows $f_1$ for given $n$ and $a$, and $f_2$ for varying values of $r$. 
If $0 < r < r_{c,1}$ or $r > r_{c,2}$, there is a single equilibrium $\bar x_1>0$, which is globally asymptotically stable because $\tfrac{dx}{dt} >0$ for $0 \leq x < \bar x_1$ and $\tfrac{dx}{dt} <0$ for $x > \bar x_1$. If $r=r_{c,1}$, there are two equilibria $0 < \bar x_1 < \bar x_2$; $\bar x_1$ is asymptotically stable with basin of attraction $[0,\bar x_2)$ and $\bar x_2$ is unstable, since $\tfrac{dx}{dt} >0$ for $0 \leq x < \bar x_1$, while $\tfrac{dx}{dt} <0$ for $\bar x_1 < x < \bar x_2$ and $x > \bar x_2$. If $r_{c,1} < r < r_{c,2}$, there are three equilibria $0 < \bar x_1 < \bar x_2 < \bar x_3$, where $\bar x_1$ and $\bar x_3$ are asymptotically stable, with basins of attraction $[0,\bar x_2)$ and $(\bar x_2, \infty)$ respectively, while $\bar x_2$ is unstable; in fact, $\tfrac{dx}{dt} >0$ for $0 \leq x < \bar x_1$, $\tfrac{dx}{dt} <0$ for $\bar x_1 < x < \bar x_2$, 
$\tfrac{dx}{dt} >0$ for $\bar x_2 < x < \bar x_3$, and $\tfrac{dx}{dt} <0$ for $x > \bar x_3$. If $r=r_{c,2}$, there are two equilibria $0 < \bar x_1 < \bar x_2$, where $\bar x_1$ is unstable and $\bar x_2$ is asymptotically stable with basin of attraction $(\bar x_1, \infty)$, since $\tfrac{dx}{dt} >0$ for $0 \leq x < \bar x_1$ and $\bar x_1 < x < \bar x_2$, while $\tfrac{dx}{dt} <0$ for $x > \bar x_2$.
The same conclusions, with $r$ replaced by $a$, can be drawn when $r$ is fixed and $a$ varies, as in Fig.~\ref{fig:intersections}b.
\end{proof}

Proposition~\ref{prop:equilibria_stability} suggests that a fold bifurcation \cite{kuznetsov2013elements, strogatz2018nonlinear} occurs when the number of equilibria changes from two to three (one equilibrium splits into two) and from three to two (two equilibria collide and merge). We consider $a$ as the bifurcation parameter and replace $f(x)$ by $f(x,a)$.

\begin{prop}\label{prop:bif}
Given $r>0$ and $n\geq2$, the system $\tfrac{dx}{dt} = f(x,a)$ undergoes fold bifurcations at the critical points $(\bar x_2(a_{c,1}),a_{c,1})$ and $(\bar x_1(a_{c,2}),a_{c,2})$, provided that $a_{c,j} \neq \frac{4n}{(n^2-1)} \sqrt[n]{\frac{n-1}{n+1}}$. Then, the system can be mapped to the normal form $\tfrac{dz}{dt} = \beta \pm z^2$ around the critical points, where $z = x - \bar x_i$ and $\beta \in \mathbb{R}$ is proportional to $a-a_{c,j}$.
$\hfill\square$
\end{prop}
\begin{proof}
For fixed $n$ and $r$, consider the values $a_{c,1}$ and $a_{c,2}$ of $a$ that satisfy \eqref{eq:equilibriumtangent}.
As shown in the proof of Proposition~\ref{prop:equilibria_stability}, when $a=a_{c,1}$ (respectively, $a=a_{c,2}$) the system admits two equilibria, $0 < \bar x_1(a_{c,1}) < \bar x_2(a_{c,1})$ with $\bar x_1(a_{c,1})$ asymptotically stable and $\bar x_2(a_{c,1})$ unstable (respectively, $0 < \bar x_1(a_{c,2}) < \bar x_2(a_{c,2})$ with $\bar x_1(a_{c,2})$ unstable and $\bar x_2(a_{c,1})$ asymptotically stable). We show that a fold bifurcation occurs both at $(\bar x_2(a_{c,1}),a_{c,1})$ and at $(\bar x_1(a_{c,2}),a_{c,2})$. According to \cite[Theorems 3.1 and 3.2]{kuznetsov2013elements}, the conditions for a fold bifurcation are: (I) $f(\bar x_i, a_{c,j})=0$; (II) $f_x(\bar x_i, a_{c,j})=0$; (III) $f_a(\bar x_i, a_{c,j}) \neq 0$; (IV) $f_{xx}(\bar x_i, a_{c,j}) \neq 0$.
In our case, (I) and (II) are satisfied by construction, since they correspond to the two conditions in \eqref{eq:equilibriumtangent}. (III) is $\frac{x^n}{1 + x^n}\vert_{x=\bar{x}_i(a_{c,j})}\neq0$, which is true, because all equilibria are strictly positive. (IV) requires that $\bar x_i \neq 0$, which is true, and $n \bar x_i^n + \bar x_i^n +1-n \neq 0$. If the latter condition is violated, $\bar x_i^n = \frac{n-1}{n+1}$ and then condition (II), $f_x(\bar x_i, a_{c,j})=\frac{a_{c,j} n \bar x_i^{n-1}}{(1+\bar x_i^n)^2}-1 = 0$, yields $a_{c,j} = \frac{4n}{(n^2-1)} \sqrt[n]{\frac{n-1}{n+1}}$, which contradicts our assumption on $a_{c,j}$. Since all conditions are verified, the system can be mapped to the normal form $\tfrac{dz}{dt} = \beta \pm z^2$ of a fold bifurcation around the critical points \cite[Section 3.3]{kuznetsov2013elements}.
\end{proof} 

Fig.~\ref{fig:bif_diag} shows bifurcation diagrams $\bar{x}(a)$ for $\tfrac{dx}{dt} = f(x,a)$ with varying $n$. The system equilibria are computed symbolically for $n=2$, numerically for larger values of $n$. Their local stability is assessed through the linearisation.

The analysis of system \eqref{eq:ODE_det} has shown that, for $n \geq 2$, it can exhibit bistability (and thus switch between ``off'' and ``on'' states in response to external stimuli, see Section~\ref{sec:Intro}), and undergoes fold bifurcations for values $a_{c,1}$ and $a_{c,2}$ of the bifurcation parameter $a$ that, for fixed $r$, depend on $n$: $a_{c,1}(n)$ and $a_{c,2}(n)$. This is a so-called ``resilience profile'' \cite{Gao2016}. As shown in the proof of Proposition~\ref{prop:equilibria_stability}, when the equilibrium $\bar x_1>0$ is unique, it is asymptotically stable with basin of attraction $B(\{\bar x_1\})=[0,\infty)$; when there are two equilibria, the asymptotically stable one is $\bar x_1(a_{c,1})$, with $B(\{\bar x_1(a_{c,1})\})=[0, \bar x_2(a_{c,1}))$, when $a=a_{c,1}$ and $\bar x_2(a_{c,2})$, with $B(\{\bar x_2(a_{c,2})\})=(\bar x_1(a_{c,2}), \infty)$, when $a=a_{c,2}$; when there are three equilibria, $\bar x_1$ and $\bar x_3$ are asymptotically stable, with $B(\{\bar x_1\})=[0,\bar x_2)$ and $B(\{\bar x_3\})=(\bar x_2,\infty)$.

We now wish to quantify the system's ability to cope with stochastic disturbances.
When two equilibria are closer (\textit{e.g.}, in the bifurcation diagram in Fig.~\ref{fig:bif_diag}, $|\bar x_3(q_1)-\bar x_2(q_1)|$ decreases when $n$ increases and $|\bar x_1(q_2) - \bar x_2(q_2)|$ decreases when $n$ decreases, with $q_i=a-a_{c,i}$),
noise-induced switches \cite{ashwin2012tipping} are intuitively more likely to occur, thus inducing transitions from the original equilibrium to another \cite{krakovska2024resilience}. To assess the likelihood of such transitions in the presence of noise, we resort to the \textit{practical resilience} framework \cite{MTNS2024}. 

\begin{figure}[t]
	\centering
	\includegraphics[width=\linewidth]{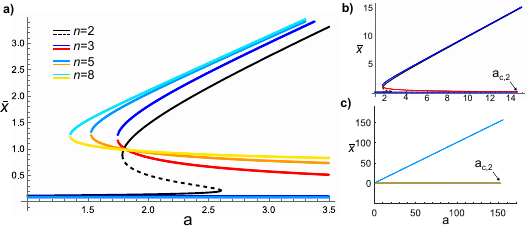}
	\caption{\textbf{(a)} Bifurcation diagram $\bar{x}(a)$ for $\tfrac{dx}{dt} = f(x,a)$, with $r=0.1$; $a_{c,1}$ is marked with the green star in each diagram, computed analytically for $n=2$ (solid and dashed black lines for equilibria $\bar{x}_1$ and $\bar{x}_3$, stable, and $\bar{x}_2$, unstable) and numerically for $n=3,5,8$ (lines in shades of blue and orange for equilibria $\bar{x}_1$ and $\bar{x}_3$, stable, and $\bar{x}_2$, unstable). For small values of $a$, decreasing $a$ enables a transition from the high to the low stable equilibrium; transitions from the low to the high stable equilibrium occur if $a$ increases, but when $n>2$ they require a huge increase in $a$, see (b,c). \textbf{(b, c)} Bifurcation diagrams for $n=3$ with $a_{c,2} \approx 15$ (b) and for $n=5$ with $a_{c,2} \approx 153$ (c). For $n=8$, $a_{c,2} > 250$ (not shown).}
	\label{fig:bif_diag}
\end{figure}

\section{Resilience of the Stochastic Model}
Given the system family $\mathcal{F}= \{G_\lambda\}_{\lambda \in \mathcal{I}}$ described by the parametric stochastic system \eqref{eq:ODE_stoch} with noise intensity $\lambda$,
consider an asymptotically stable equilibrium $A$ (attractor) of the \emph{nominal} deterministic system $G_{\lambda_0}$ given by \eqref{eq:ODE_det}, and the corresponding basin of attraction $B(A)$.
To quantify the probability that, despite the stochastic noise, the system trajectories emanating from the initial condition $x_0 \in B(A)$ remain within a prescribed neighbourhood of $A$ for the whole time interval $t \in [0, \tau)$, the concept of \textit{practical resilience} has been introduced in \cite[Definition 2]{MTNS2024}, formalising heuristic notions of resilience employed in systems biology \cite{Dai2015}.
\begin{definition}\label{def:res}
\cite{MTNS2024} Consider the system family $\mathcal{F} = \left\{G_{\lambda} \right\}_{\lambda \in \mathcal{I}}$ and let $(A,B(A))$ be an attractor-basin pair corresponding to $G_{\lambda_0}$. Fix the time horizon $\tau\in (0,\infty)$, distance $\delta \geq 0$ and confidence level $\gamma \in (0,1]$. Consider the set $A_\varepsilon = \{ \chi \colon \operatorname{dist}\left(\chi, A \right) \leq \varepsilon \}$, with $0 \leq \varepsilon \leq \delta$. The system $G_{\lambda}$ is $(\tau,\gamma,\delta, \varepsilon)$-\emph{practically resilient} if, for all $x_0\in A_\varepsilon \cap B(A)$,
\begin{equation}\label{eq:pract_res}
\mathbf{P}_{\lambda} \doteq \mathbb{P}_{\lambda}\left( \sup_{ t\in [0,\tau)} \operatorname{dist}\left({x}(t;x_0,\eta_{\lambda}), A \right) \leq \delta \right)\geq \gamma.
\end{equation}
The family $\mathcal{F}$ is $(\tau,\gamma,\delta,\varepsilon)$-\emph{practically resilient} if, for all $x_0\in A_\varepsilon \cap B(A)$,
$\inf_{\lambda\in \mathcal{I}\setminus \left\{\lambda_0 \right\}} \mathbf{P}_{\lambda} \geq \gamma$.
$\hfill\diamond$
\end{definition}
System $G_{\lambda}$, $\lambda\in \mathcal{I}\setminus \left\{\lambda_0 \right\}$, is $(\tau,\gamma,\delta, \varepsilon)$-practically resilient if, subject to the stochastic noise $\eta_{\lambda}$, the trajectory $x(t;x_0,\eta_{\lambda})$, emanating from an arbitrary point $x_0$ that lies within an $\varepsilon$-distance from the attractor $A$, and within its basin of attraction $B(A)$, remains within a $\delta$-distance from $A$ with probability at least $\gamma$ over the interval $t\in [0,\tau)$ \cite{MTNS2024}.
Due to the stochastic noise in the perturbed system, the system trajectories emanating from $B(A)$ may not converge to $A$ (which is consistent with persistent noise-induced fluctuations observed in experiments), but the definition requires them to remain within a neighbourhood of the attractor $A$ of the nominal system, with high enough probability, for all $t \in [0, \tau)$. 
We can study how the probability of preserving a prescribed attractor-basin pair depends on the system parameters; by fixing $\gamma$, we can identify the system parameters that yield a desired confidence level.

\textbf{Numerical simulations.} To quantify $\mathbf{P}_\lambda$ numerically for an attractor $\bar x_i$, \textit{i.e.}, an asymptotically stable equilibrium of \eqref{eq:ODE_det}, we generate 200 random trajectories $x(t;x_0,\eta_\lambda)$ with $x_0 \in B\left(\left\{\bar x_i\right\}\right)$, and assess how frequently they remain within $(\bar x_i - \delta, \bar x_i+\delta)$ for the whole horizon $[0, \tau)$. We choose the value $\tau=100$ such that the trajectories of the deterministic system \eqref{eq:ODE_det} have already reached their steady state.
In particular, we perform numerical simulations of system \eqref{eq:ODE_stoch} using an Euler-Maruyama scheme with $\Delta t=0.01$, using $r=0.1$ and $a(n)=a_{c,1}(n)+0.05$, which guarantees bistability for all the considered choices of $n \in \{2,3,5,8\}$ (\textit{cf.} Fig.~\ref{fig:bif_diag}), and with uniformly spaced initial conditions $x_0 \in B(\{\bar x_3(n)\})=(\bar x_2(n), r+a(n)]$, where $r+a(n)$ is the upper limit of the invariant set for system \eqref{eq:ODE_det} identified in Section~\ref{sec:det:mod}. We consider values of the bifurcation parameter $a$ close to $a_{c,1}$ because the transition from ``on'' to ``off'' states is particularly interesting in many biological systems, such as in the case of $\beta$-\textit{galactosidase} expression \cite{ozbudak2004multistability, proverbio2022buffering}. 

For the various choices of $n$ and for a range of noise intensities $\lambda$, we compute the fraction of trajectories of \eqref{eq:ODE_stoch} that remain within $(\bar x_3(n) - \delta, \bar x_3(n)+\delta)$ for the whole horizon $[0,\tau)$.
We set $\delta = 0.25$, so that $|\bar{x}_3(n) - \bar{x}_2(n)| < \delta$ and $|r+a(n) - \bar{x}_3(n)| < \delta$ for all considered values of $n$. 
Fig.~\ref{fig:sims}a shows $\mathbf{P}_\lambda$ depending on $n$ and $\lambda$: the probability of remaining within a $\delta$-neighbourhood of the attractor is larger if the noise intensity $\lambda$ is smaller (as expected). If $n$ increases, $\mathbf{P}_\lambda$ first increases and then decreases: more complex oligomers (formed by a larger number $n$ of monomers) initially increase the resilience to noise, but then bring the system to lower levels of resilience (linked to the small peaks of $\hat{p}(x)$ for $\bar{x}_3(n)$, which will be shown in Fig.~\ref{fig:p_s}). The non-monotonic behaviour of the resilience with $n$ suggests that we can identify a value of $n$ for which the system achieves maximal resilience; this is a promising avenue for future studies in synthetic biology.

For a fixed $\lambda=0.01$, Fig.~\ref{fig:sims}b shows $\mathbf{P}_{\lambda}$ for different values of $\delta \in (0,0.6]$, where a smaller $\delta$ corresponds to a stricter requirement. For smaller values of $n$, $\mathbf{P}_\lambda$ is higher for larger $\delta$; conversely, given a fixed probability $\gamma^*$, the smallest $\delta$-neighbourhood of $\bar x_3(n)$ for which $\mathbf{P}_{\lambda}$ is at least $\gamma^*$ is smaller when $n$ is larger (more complex oligomers), which is favourable to reduce variability in the concentration of TFs, thus confirming the qualitative observations in \cite{proverbio2022buffering}.

\begin{figure}[htb]
	\centering
	\includegraphics[width=0.7\linewidth]{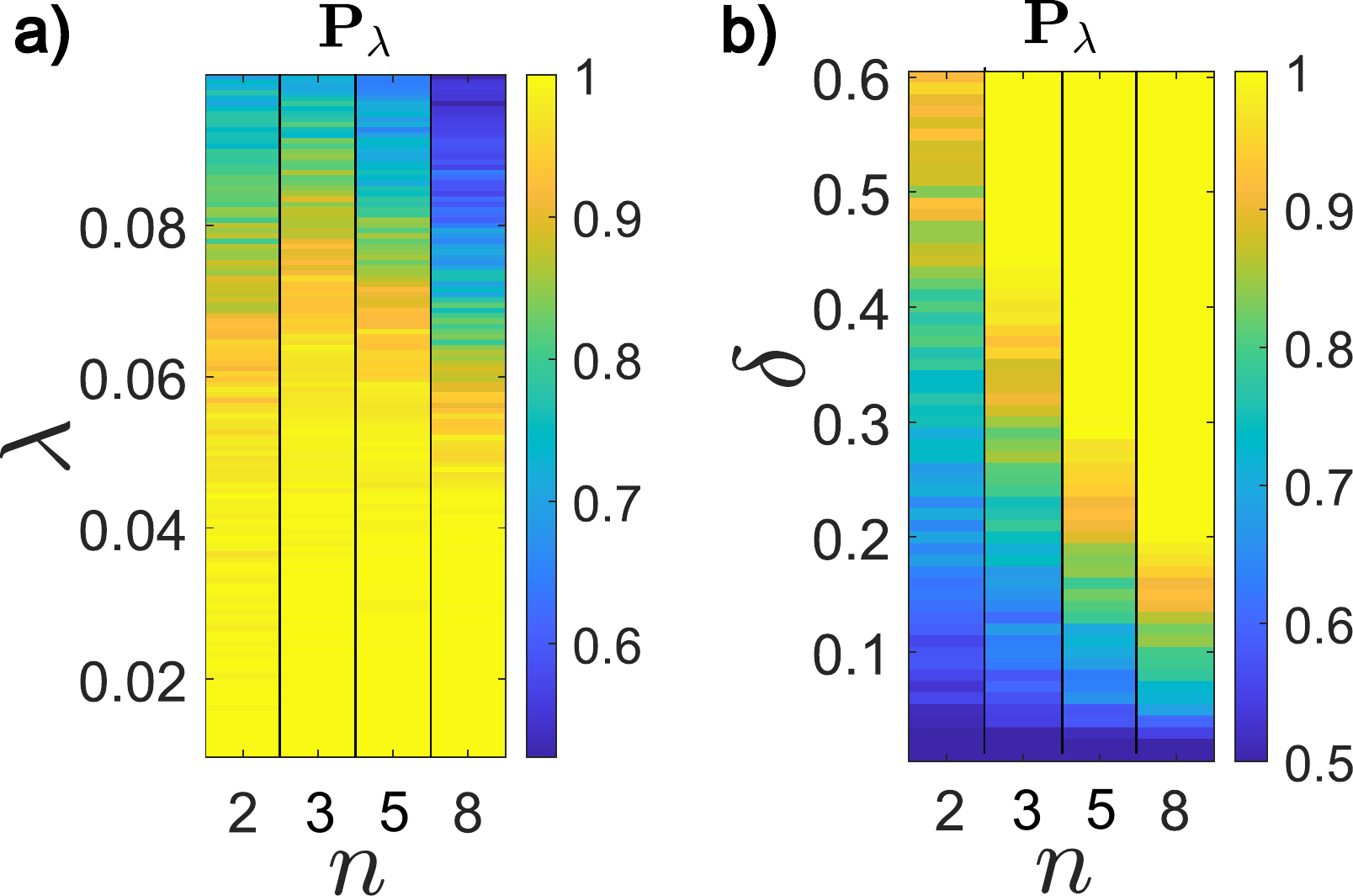}
	\caption{\textbf{(a)} Dependence of $\mathbf{P}_\lambda$ on $\lambda$ and $n$, for system \eqref{eq:ODE_stoch} with $r=0.1$, $a=a_{c,1}(n)+0.05$ and a fixed $\delta = 0.25$. \textbf{(b)} Dependence of $\mathbf{P}_\lambda$ on $\delta$ and $n$, for system \eqref{eq:ODE_stoch} with $r=0.1$, $a=a_{c,1}(n)+0.05$ and a fixed $\lambda = 0.01$.}
	\label{fig:sims}
\end{figure}

\section{Fokker-Planck Quantification of Resilience}
System \eqref{eq:ODE_stoch} describes the evolution of a stochastic process, driven by white noise, having probability density function (PDF) $p(x,t)$.
Here, we offer a complementary definition of resilience of attractors, based on properties of the probability distribution of the stochastic process. In particular, to assess the probability that a trajectory of \eqref{eq:ODE_stoch}, computed with the It\^{o} formalism, lies in a prescribed set at a given time, we consider the corresponding Fokker-Planck equation (FPE), which describes the evolution of $p(x,t)$. In fact, any stochastic process whose PDF $p(x,t)$ satisfies the FPE associated with the SDE \eqref{eq:ODE_stoch} is equivalent to the It\^{o} solution of  the SDE \cite{Gardiner1985}. Following \cite{Gardiner1985}, a generic SDE in the Langevin equation form $dx = a(x,t) dt + \sqrt{b(x,t)} \tilde \eta(t)dt$,
where the integral of the stochastic noise $\tilde \eta$ is a Wiener process, has an associated FPE for $p(x,t)$:
\begin{equation*}
    p_t(x,t) = - [a(x,t)p(x,t)]_x+\tfrac{[b(x,t)p(x,t)]_{xx} }{2}= -J_x(x,t),
\end{equation*}
where we have introduced the probability current
\begin{equation}\label{eq:J}
J(x,t) = a(x,t)p(x,t) -\tfrac{1}{2}[b(x,t)p(x,t)]_x
\end{equation}
and the subscript $[\cdot]_z$ denotes partial differentiation in $z$.
The FPE associated with \eqref{eq:ODE_stoch} in the It\^{o} formalism is  \begin{equation}\label{eq:FPE_loop}
    p_t(x,t) = - \left[f(x)p(x,t) \right]_x + \frac{\lambda^2}{2} p_{xx}(x,t),
\end{equation}
where $a(x,t)=f(x)$ and $\sqrt{b(x,t)} = \lambda$ do not depend on time.
We set reflecting boundary conditions for $J(x,t)$ at $x=0$ and $x=+\infty$: $J(0,t) = J(+\infty,t) = 0$, for all $t$.

The following resilience definition considers the probability that a realisation of the perturbed system trajectory with initial condition $x_0$ lies in a $\delta$-neighbourhood of the nominal attractor $A$ at time $t$, for all times $t > T$, with $T>0$.

\begin{definition}\label{def:resFP}
Consider the system family $\mathcal{F} = \left\{G_{\lambda} \right\}_{\lambda \in \mathcal{I}}$ and let $(A,B(A))$ be an attractor-basin pair corresponding to $G_{\lambda_0}$. Fix the distance $\delta \geq 0$ and confidence level $\gamma \in (0,1]$, and define the event $\mathfrak{R}_{t,\delta} \doteq \{x(t;x_0,\eta_{\lambda}) \in A_\delta \}$. The attractor $A$ is $(\gamma,\delta)$-\emph{resilient} for the system $G_{\lambda}$ if there exists $T \in (0,\infty)$ such that for all $t > T$, for all $x_0\in B(A)$,
\begin{equation}\label{eq:resFP}
\mathbb{P}_\lambda(\mathfrak{R}_{t,\delta}) = \int_{A_\delta} p(x,t) dx \geq \gamma.
\end{equation}
The attractor $A$ is $(\gamma,\delta)$-\emph{resilient} for the family $\mathcal{F}$ if, for all $x_0\in B(A)$,
$\inf_{\lambda\in \mathcal{I}\setminus \left\{\lambda_0 \right\}} \mathbb{P}_\lambda(\mathfrak{R}_{t,\delta}) \geq \gamma$. $\hfill\diamond$
\end{definition}

For system \eqref{eq:ODE_stoch}, the probability of $x(t;x_0,\eta_{\lambda})$ lying in a $\delta$-neighbourhood of the attractor $\bar x_i$ (at time $t$) is $\mathbb{P}_\lambda(\mathfrak{R}_{t,\delta}) = \int_{\bar x_i-\delta}^{\bar x_i+\delta} p(x,t) dx$.
In the limit for large enough times, such a probability can be approximated by
\begin{equation}
\label{eq:cum_prob}
    \mathcal{P}_\lambda \doteq \int_{\bar x_i-\delta}^{\bar x_i+\delta} \hat{p}(\xi) d\xi,
\end{equation}
where $\hat{p}(x)$ is the time-independent \textit{stationary solution} of \eqref{eq:FPE_loop} and, under suitable assumptions, is the limit of the solutions $p(x,t)$ for $t\to\infty$ in the $L^1$ metric
(see, \textit{e.g.}, \cite[Sections 5.4 and 6.1]{Risken1984} and \cite{CarrilloToscani1998} for more details). 
Hence, $\hat{p}(x)$ captures the asymptotic PDF of the stochastic process $x$ solving \eqref{eq:ODE_stoch} in the It\^{o} sense.

Following \cite[Chapter 5.2.2]{Gardiner1985}, to obtain $\hat{p}(x)$, we set $p_t(x,t)=0$ in \eqref{eq:FPE_loop}. Since $J_x(x,t)=J_x(x)=0$ and the boundary conditions are zero, this amounts to setting $J = 0$ in \eqref{eq:J}, with $a(x,t)=f(x)$ and $\sqrt{b(x,t)} = \lambda$, namely
\begin{equation}\label{eq:Psequation}
    f(x)\hat{p}(x) = \tfrac{\lambda^2}{2} \tfrac{d\hat{p}}{dx}(x)=0,
\end{equation}

which yields
\begin{equation}\label{eq:Ps}
    \hat{p}(x) = \tfrac{\nu}{\lambda^2} \exp\left[ \tfrac{2}{\lambda^2}\int_{0}^x f(\xi)d\xi \right],
\end{equation}
where $\nu$ is a normalization constant s.t. $\int_{0}^{\infty} \hat{p}(x) dx = 1$.

Since $f(x)$ is integrable, $\hat{p}(x)$ can be rewritten as
\begin{equation}
\label{eq:ps_explicit}
     \hat{p}(x) = \tfrac{\nu}{\lambda^2} \exp \left[\tfrac{1}{\lambda^2}\left(2 r x - {x^2} + \tfrac{2 a x^{n + 1} \, \mathfrak{F}(x;n)}{n + 1}\right) \right] \ ,
\end{equation}
where $\mathfrak{F}(x;n)={_2F_1}\left[1, \frac{n+1}{n}; \frac{2n+1}{n}; -x^n\right]$ is the Gaussian hypergeometric function \cite[Chapter 15]{abramowitz1965handbook}.
Therefore, $\hat{p}(x)$ depends explicitly on $n$, $r$ and $a$. We study how it depends on the bifurcation parameter $a$ considered in Proposition~\ref{prop:bif}, also in relation to the equilibria of the nominal deterministic system analysed in Proposition~\ref{prop:equilibria_stability}.

\begin{prop}\label{prop:PDF}
For $n\geq 2$, $\hat{p}(x)$ has two local maxima when $a_{c,1}(n) < a < a_{c,2}(n)$ and a single maximum when $0 < a \leq a_{c,1}(n)$ or $a \geq a_{c,2}(n)$. Moreover, its maxima are achieved for values of $x$ corresponding to the asymptotically stable equilibria of the nominal deterministic system \eqref{eq:ODE_det}. 
$\hfill\square$
\end{prop}
\begin{proof}
Denoting by $h(x)$ the argument of the exponential in \eqref{eq:Ps} and substituting the expression of $\hat{p}(x)$ into \eqref{eq:Psequation} yields
$\frac{d\hat{p}}{dx}(x) = \frac{2}{\lambda^2} f(x) \frac{\nu}{\lambda^2} \exp[h(x)]=0$.
Since the exponential function is always positive, this is equivalent to requiring $f(x)=0$, which is the equilibrium condition for system \eqref{eq:ODE_det}. Hence, the statement follows from the results in Section~\ref{sec:det:mod}. The stationary points of $\hat{p}(x)$ are attained for values of $x$ that are equilibria of \eqref{eq:ODE_det}. Considering the sign of $f(x)$ within the intervals of values of $x$ delimited by the zeros of $f(x)$, we have a maximum for $\hat{p}(\bar x_i)$ if $\bar x_i$ is asymptotically stable, a minimum for $\hat{p}(\bar x_i)$ if $\bar x_i$ is the unstable equilibrium in the three-equilibria case, and an inflection point for $\hat{p}(\bar x_i)$ if $\bar x_i$ is the unstable equilibrium in the two-equilibria case. Hence, $\hat{p}(x)$ has two local maxima in the three-equilibrium case and a single maximum otherwise.
\end{proof}

In view of Proposition~\ref{prop:PDF}, $\hat{p}(x)$ is a good indicator of the system resilience in relation to the preservation of prescribed attractors: its maxima occur \textit{at} the attractors, and the maxima are two when the deterministic system is bistable, because noise can drive the stochastic trajectories to a different basin of attraction, with a probability that increases with the noise intensity $\lambda$. 

Fig.~\ref{fig:p_s} shows $\hat{p}(x)$ as in \eqref{eq:ps_explicit} for different values of ${n}$, for increasing noise intensity $\lambda$, for three different values of $a$. For $a=0.8$ and $a=2.7$, system \eqref{eq:ODE_det} is monostable for all considered $n$ and its equilibrium is at low values (``off'') and high values (``on''), respectively; for $a = a_{c,1}(n) + 0.05$, the system is bistable. When $\lambda^2=0.01$, in the monostable cases, all the noisy trajectories lie very close to the attractor of the nominal deterministic system: $\hat{p}(x)$ is  an approximation of a Dirac delta centred at the attractor when $a=2.7$, and is centred at the attractor with a very narrow peak when $a=0.8$; in the bistable case ($a = a_{c,1}(n) + 0.05$), $\hat{p}(x)$ has a maximum at the low asymptotically stable equilibrium, a minimum at the unstable equilibrium and a second maximum at the high asymptotically stable equilibrium. When $\lambda^2=0.05$, the trajectories are more likely to lie within regions of the state space further away from the attractors. For all noise intensities, when $n$ is larger, $\hat{p}(x)$ has higher and narrower peaks centred at the equilibria, thus confirming the observation in \cite{proverbio2022buffering} that noise suppression close to the equilibria improves when using more complex oligomers (composed of a larger number $n$ of monomers).

\begin{figure}[t]
	\centering
	\includegraphics[width=0.95\linewidth]{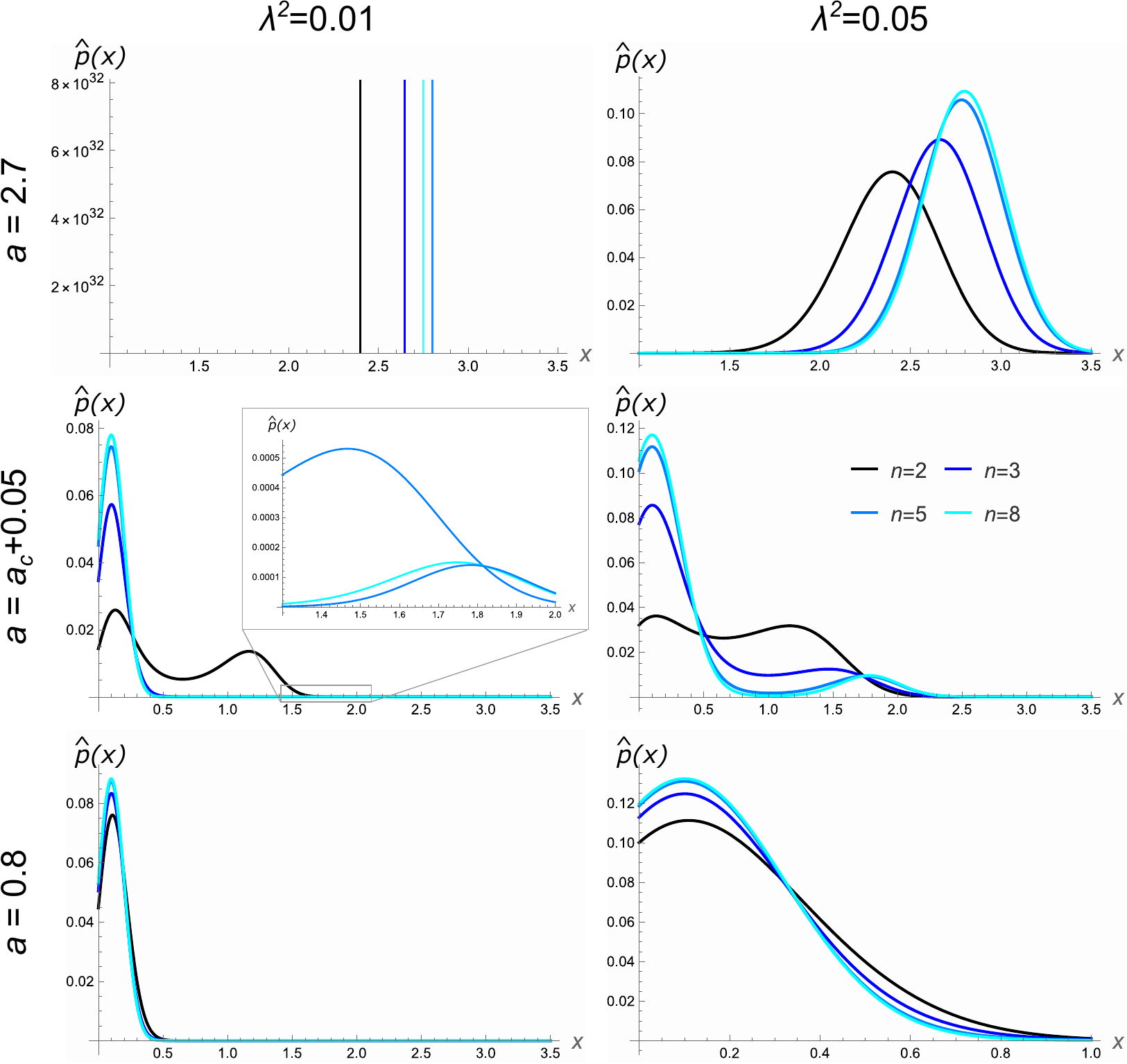}
	\caption{Plots of $\hat{p}(x)$ from \eqref{eq:ps_explicit}, for $r=0.1$ and $n \in \{2,3,5,8\}$. Columns: $\lambda^2 = 0.01$ (left) and $\lambda^2 = 0.05$ (right). Rows: $a=2.7$ (top); $a = a_{c,1}(n) + 0.05$ (middle); $a = 0.8$ (bottom). The constant $\nu$ is estimated numerically for normalization. The inset on the second row for $\lambda^2=0.01$ shows the second maximum for n>2.}
	\label{fig:p_s}
\end{figure}


\section{Conclusion}

We have investigated the resilience of the positive gene autoregulating feedback loop motif, where feedback is enforced by oligomers formed by $n$ monomers. The deterministic system can exhibit both monostable and bistable behaviours depending on the parameter values, and undergoes fold bifurcations when crucial parameters are varied. First, we have assessed \textit{practical resilience} \cite{MTNS2024} of the system in the presence of noise, and analysed how different parameters affect the system's resilience, focusing in particular on the oligomer size $n$, for which there seems to exist an optimal value that maximises resilience.
Then, we have introduced a complementary definition of resilience in a stochastic framework that relies on the Fokker-Planck equation, thus bridging bistability in the deterministic case with a bimodal stationary distribution in the stochastic case.
Modelling biological stochasticity as white noise is a common first approximation \cite{o2018stochasticity} to account for various sources of uncertainty in transcription; more complex models, such as stochastic chemical kinetics models, could be considered to mechanistically capture randomness in transport and binding of TFs, as well as transcriptional noise. 

\section*{Acknowledgments}
We are grateful to Rami Katz for precious help and discussions on the topic of this manuscript, and to the anonymous reviewers for their valuable comments.

This work was funded by the European Union through the ERC INSPIRE grant (project number 101076926). Views and opinions expressed are however those of the authors only and do not necessarily reflect those of the European Union or the European Research Council Executive Agency. Neither the European Union nor the European Research Council Executive Agency can be held responsible for them.

\section*{Code}
The code to reproduce the results is at\\ \url{https://github.com/daniele-proverbio/GeneAutocatalysis}.

\bibliography{sample.bib}

\end{document}